\documentclass{lmcs}
\pdfoutput=1

\usepackage{lastpage}
\lmcsdoi{20}{1}{10}
\lmcsheading{}{\pageref{LastPage}}{}{}%
{Dec.~26,~2022}{Jan.~29,~2024}{}

\keywords{Games on graphs, Memory}

\usepackage{hyperref}
\usepackage[norelsize,ruled,vlined]{algorithm2e}
\usepackage[utf8]{inputenc}
\usepackage{color}
\usepackage{graphicx}
\usepackage{subcaption}

\title{Playing Safe, Ten Years Later}
\titlecomment{A preliminary version of this work was published in the proceedings of FSTTCS'2014~\cite{ColcombetFH14}.}

\author[T. Colcombet]{Thomas Colcombet\lmcsorcid{0000-0001-6529-6963}}[a]
\address{CNRS, IRIF, Universit\'e de Paris, Paris, France}

\author[N. Fijalkow]{Nathana{\"e}l Fijalkow\lmcsorcid{0000-0002-6576-4680}}[b]
\address{CNRS, LaBRI, Bordeaux, France and University of Warsaw, Warsaw, Poland}

\author[F. Horn]{Florian Horn\lmcsorcid{0000-0001-8872-4705}}[a]

\usepackage{amssymb,amsthm}
\usepackage{framed}
\usepackage{tikz}

\usetikzlibrary{shapes,shapes.symbols,automata,arrows}
\tikzset{>=stealth, shorten >=1pt}
\tikzset{every edge/.style = {thick, ->, draw}}
\tikzset{Eve/.style = {shape = ellipse, draw, thick, minimum size = .6cm, fill=gray!30!white}}
\tikzset{Adam/.style = {shape = rectangle, minimum size =.6cm, draw, fill=gray!30!white}}

\newcommand{\set}[1]{\{ #1 \}}

\newcommand{\Safe}{\mathrm{Safe}}

\newcommand{\Res}{\mathrm{Res}}
\renewcommand{\next}{\nu}
\newcommand{\mem}{\mathrm{mem}}
\newcommand{\memFD}{\mathrm{mem}_{\textrm{fin-deg}}}
\newcommand{\M}{\mathcal{M}}
\newcommand{\up}{\mu}
\newcommand{\N}{\mathbb{N}}

\newcommand{\VE}{V_\exists}
\newcommand{\VA}{V_\forall}
\newcommand{\W}{\mathcal{W}}
\newcommand{\A}{\mathcal{A}}
\newcommand{\WE}{\W_{E}}
\newcommand{\WA}{\W_{A}}
\newcommand{\G}{\mathcal{G}}

\theoremstyle{thmC}
\newtheorem{corC}[thm]{Corollary}

\begin{document}

\begin{abstract}
We consider two-player games over graphs and give tight bounds on the memory size of strategies ensuring safety objectives. More specifically, we show that the minimal number of memory states of a strategy ensuring a safety objective is given by the size of the maximal antichain of left quotients with respect to language inclusion. This result holds for all safety objectives without any regularity assumptions. We give several applications of this general principle. In particular, we characterize the exact memory requirements for the opponent in generalized reachability games, and we prove the existence of positional strategies in games with counters.
\end{abstract}

\maketitle
\section{Introduction}
Graphs games provide a mathematical framework to model reactive systems (we refer to~\cite{GTW02-Dagstuhl} for a survey on the topic, and to~\cite{GamesBook} for a recent textbook).
We focus here on the Synthesis Problem to motivate the problem we consider, which is to characterize the amount of memory required in games with safety objectives.

\medskip\noindent{\bf The synthesis problem.}
The inputs of the synthesis problem are a \textit{system} and a \textit{specification}.
The expected output is a \textit{controller} for the system, that ensures the specification.
We describe here a classical and well-studied approach to solve the synthesis problem through game theory.
We model the system as a graph, whose vertices represent states and edges represent transitions.
Its evolution consists in interactions between a controller and an environment,
which is turned into a game on the graph between two players, Eve and Adam.
If in a given state, the controller can choose the evolution of the system,
then the corresponding vertex is controlled by Eve.
If the system evolves in an uncertain way, we consider the worst-case scenario,
where Adam controls those states.
A pebble is initially placed on the vertex representing 
the initial state of the system, 
then Eve and Adam move this pebble along the edges.
The sequence built describes a run of the system: 
Eve tries to ensure that it satisfies the specification.
So, in order to synthesize a controller, we are interested in 
whether Eve can ensure this objective and what resources she needs.
In particular, the most salient question is: what is the size of a \textit{minimal} controller
satisfying the specification?
Since a controller is here given by a strategy for Eve, this is equivalent to the following question:
what is the minimal amount of memory used by a winning strategy?
The following diagram shows the correspondence between the notions from the synthesis problem (left-hand side) and the game-theoretic notions (right-hand side).
$$
\underbrace{\mathbb{S}}_{\textrm{system}},\underbrace{\mathcal{C}}_{\textrm{controller}} 
\models \underbrace{\Phi}_{\textrm{specification}}
\qquad \Longleftrightarrow \qquad 
\underbrace{\A}_{\textrm{arena (graph)}},\underbrace{\sigma}_{\textrm{strategy}} 
\models \underbrace{W}_{\textrm{objective}}
$$

\medskip\noindent{\bf Safety specifications.}
Since we consider non-terminating sequences, a specification is given by a language of infinite words.
Of special interest are the specifications asserting that ``nothing bad'' ever happens; such specifications are called \textit{safety specifications}.
A safety specification is induced by a (possibly infinite and non-regular) set of bad prefixes $P$,
and the specification is met by a run if no prefixes belong to $P$.
Although quite simple, the safety specifications proved useful in both theory and practice, and are actively studied (we refer to~\cite{Kupferman14} for a survey).

\medskip\noindent{\bf Our contribution.}
In this paper, our goal is to characterize the memory requirements of (arbitrary) safety winning objectives.
The reader with a background in game theory may be surprised, as it is well known that ``safety games are positionally determined'', implying that the quantity above is constant equal to one.
The subtlety is that our setting is (much) more general than the classical notion of safety games: 
in a safety game, the goal is to avoid a set of forbidden edges, while a safety objective defines a set of forbidden prefixes independent of the arena.
We show the following general principle:
\begin{framed}
\centering
For a safety winning objective $W$, the minimal number of memory states of a winning strategy
is exactly the cardinal of the maximal antichain of left quotients of $W$.
\end{framed}
We refer to Section~\ref{sec:safety} for the missing definitions.
Note that this result holds for all safety winning objectives, without any regularity assumption.
The characterisation above holds for graphs with finite degree, however we can state and prove a variant lifting this assumption, but assuming that the set of left quotients is well founded with respect to inclusion.
We state our main results in Section~\ref{sec:statements}, and prove them in Section~\ref{sec:safety}.
We give several examples and applications in Section~\ref{sec:applications}.
For instance, it allows to characterize the memory requirements for the opponent in generalized reachability games,
and to prove the existence of positional strategies in games with counters.
This journal version additionally discusses related works in Section~\ref{sec:related_works}, many of them posterior to the publication of the conference version of this article~\cite{ColcombetFH14} in 2014.

\section{Definitions}
\label{sec:defs}
The games we consider are played on an \emph{arena} $\A = (V,(\VE,\VA),E, c)$,
consisting of a (finite or infinite) graph $(V,E)$, a partition $(\VE,\VA)$ of the vertex set $V$: 
a vertex in $\VE$ belongs to Eve and in $\VA$ to Adam, and a coloring function $c : E \to A$ mapping edges to a color from a finite alphabet $A$.
When drawing arenas, we will use circles for vertices owned by Eve and squares for those owned by Adam.
Throughout this paper, we make the cosmetic assumption that graphs do not have dead-ends: for every vertex $v \in V$, there exists an edge $(v,v') \in E$.

\vskip1em
\textbf{Game.} A \emph{play} $\pi$ is an infinite word of edges $e_0 \cdot e_1 \cdots$ that are consecutive: for all $i$, 
$e_i = (\_,v) \in E$ and $e_{i+1} = (v,\_) \in E$ for some $v \in V$.
The prefix of length $k$ of $\pi$ is denoted $\pi_k$.
A play $\pi$ induces an infinite sequence of colors $c(\pi)$, obtaining by applying the coloring function $c$ component-wise.
We define \emph{winning objectives} for a player by giving a set of infinite sequences of colors $W \subseteq A^\omega$. 
As we are interested in zero-sum games, \textit{i.e.} where the winning objectives of the two players are opposite, 
if the winning objective for Eve is~$W$, then the winning objective for Adam is $A^\omega \setminus W$.
A \emph{game} is a couple $\G = (\A,W)$ where $\A$ is an arena and $W$ a winning objective.

\vskip1em
\textbf{Strategy.} A \emph{strategy} for a player is a function that prescribes, given a finite history of the play, the next move. 
Formally, a strategy for Eve is a function $\sigma : E^* \cdot \VE \to E$ such that for all $\pi \in E^*$ and $v \in \VE$ we have $\sigma(\pi \cdot v) = (v,\_) \in E$.
Strategies for Adam are defined similarly, and usually denoted $\tau$.
Once a game $\G = (\A,W)$, a starting vertex $v_0$ and strategies $\sigma$ for Eve and $\tau$ for Adam are fixed, 
there is a unique play $\pi(v_0,\sigma,\tau)$, which is said winning for Eve if its image by $c$ belongs to $W$.

A strategy $\sigma$ for Eve is \emph{winning} if for all strategies $\tau$ for Adam, $\pi(q_0,\sigma,\tau)$ is winning.
We say that Eve wins the game $\G$ from $v_0$ if she has a winning strategy from $v_0$, 
and denote $\WE(\G)$ the set of vertices from where Eve wins;
we often say that $v \in \WE(\G)$ is winning.
We define similarly $\WA(\G)$ for Adam to be the set of vertices from where Adam wins.

\vskip1em
\textbf{Memory.} A \emph{memory structure} is a deterministic state machine that reads the sequence
of edges and abstracts its relevant informations into a memory state.
Formally, a memory structure $\M = (M, m_0, \up)$ for an arena consists of a set $M$ of memory states, 
an initial memory state $m_0 \in M$ and an update function $\up: M \times E \to M$.
The update function takes as input the current memory state and the chosen edge to compute the next memory state.
It can be extended to a function $\up^*: E^* \to M$ by defining $\up^*(\varepsilon) = m_0$ and $\up^* (\pi \cdot e) = \up(\up^*(\pi), e)$.
Given a memory structure $\M$ and a next-move function $\next: \VE \times M \to E$, we can define a strategy $\sigma$ for Eve by $\sigma(\pi \cdot v) = \next(v, \up^*(\pi \cdot v))$.
A strategy with memory structure $\M$ has finite memory if $M$ is a finite set, and we write $|\M|$ for the size of $M$.
It is \emph{memoryless}, or \emph{positional} if $M$ is a singleton: it only depends on the current vertex. Hence a memoryless strategy can be described as a function $\sigma: \VE \to E$.

\vskip1em
An arena and a memory structure induce an expanded arena where the current memory state is computed online.
Formally, the arena $\A = (V, (\VE,\VA), E, c)$, the memory structure $\M$ for $\A$ 
and a new coloring function $c' : E \times M \to A$ induce an expanded arena 
$\A \times \M = (V \times M, (\VE \times M, \VA \times M), E \times \up, c')$, where
$E \times \up$ is defined by: $((v,m), (v',m')) \in E \times \up$ if $(v,v') \in E$ and $\up(m,(v,v')) = m'$.

From a memoryless strategy in $\A \times \M$, we can build a strategy in $\A$ using $\M$ as memory structure, which behaves as the original strategy.
This key observation will be used several times in the paper.

\section{Statements of the Results}
\label{sec:statements}

In this section, we consider a safety objective\footnote{To be defined in this section.} $W$ and compute the following quantity:
$$\mem(W)\ \doteq\ \sup_{\begin{subarray}{c}\G = (\A,W) \textrm{ game}\\ v_0 \textrm{ initial vertex}\end{subarray}}\ 
\inf_{\begin{subarray}{c}\sigma \textrm{ winning strategy from $v_0$}\\ \textrm{with memory structure $\M$}\end{subarray}}\ |\M|\ .$$
In words, $\mem(W)$ is the necessary and sufficient number of memory states for constructing
a winning strategy in games with objective $W$.
Equivalently:
\begin{itemize}
	\item \textit{upper bound:} for all games $\G = (\A,W)$, if Eve has a winning strategy from an initial vertex $v_0$,
	then she has a winning strategy using at most $\mem(W)$ memory states,
	\item \textit{lower bound:} there exists a game $\G = (\A,W)$ and an initial vertex $v_0$ where Eve has a winning strategy,
	but no winning strategy using less than $\mem(W)$ memory states.
\end{itemize}
Note that we place no restrictions on the size of the games involved: they may be infinite.
We explain now why our study of games with safety objectives is (much) more general than the classical notion of safety games.

Consider an arena $\A$: a safety \textit{condition} is given by a subset $B \subseteq E$ of forbidden edges, 
inducing the winning condition 
$$\set{\rho = a_0 \cdot a_1 \cdots \in E^\omega \mid \textrm{for all } i,\ a_i \notin B}\ .$$

Let $A$ a set of colours. A safety \textit{objective} is given by a subset $P \subseteq A^*$ of forbidden prefixes of colors,
inducing the winning objective 
$$\Safe(P) = \set{\rho = a_0 \cdot a_1 \cdots \in A^\omega \mid \textrm{for all } i,\ a_0 \cdot a_1 \cdots a_i \notin P}\ .$$
For each arena $\A$ equipped with a colouring function $c : E \to A$, this induces the winning condition
$$\set{\rho = a_0 \cdot a_1 \cdots \in E^\omega \mid c(\rho) = c(a_0) \cdot c(a_1) \cdots \in \Safe(P)}.$$

\vskip1em
In other words, a condition is defined directly on an arena, while an objective is defined independently of the arena on the sequence of colours and induces a condition for each arena equipped with a colouring function.
As an example, consider the safety objective $\Safe(P)$ over the colours $A = \set{a,b,c}$ 
where $P = A^* \cdot a \cdot A^* \cdot b \cup A^* \cdot b \cdot A^* \cdot a$: the objective $\Safe(P)$ specifies that $a$ and $b$ cannot be both seen along a play.
This cannot be expressed by a safety condition, and $\mem(\Safe(P)) \ge 2$, as shown in Figure~\ref{fig:example_memory}.

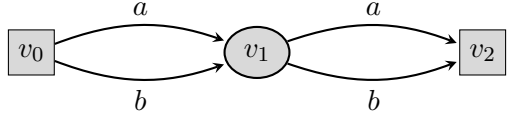
\begin{figure}
\begin{center}
\begin{tikzpicture}[scale = 1]
\begin{scope}
\node[Adam]  at (0,2) (v0)	{$v_0$};
\node[Eve] at (3,2)	 (v1)	{$v_1$};
\node[Adam]  at (6,2) (v2)	{$v_2$};
\end{scope}

\path[thick]
(v0) edge[bend left = 20] node[above] {$a$} (v1)
(v0) edge[bend right = 20] node[below] {$b$} (v1) 

(v1) edge[bend left = 20] node[above] {$a$} (v2)
(v1) edge[bend right = 20] node[below] {$b$} (v2) 
;
\end{tikzpicture}
\end{center}
\caption{Memory is necessary for safety objectives: to avoid seeing both $a$ and $b$, Eve must choose the same letter as Adam did in the first move, which requires two memory states.}
\label{fig:example_memory}
\end{figure}

\begin{lem}[Folklore]\label{lem:memoryless_internal_safety}
Let $\G$ be a game with a safety condition, and $v_0$ an initial vertex.
If Eve has a winning strategy, then she has a positional winning strategy.
\end{lem}

Safety objectives form a very expressive class of objectives, as we will demonstrate in Section~\ref{sec:applications}.
The notion of safety objective originates from topological studies of the set of infinite words:
the safety objectives are the closed sets for the Cantor topology, denoted $\Sigma_1$ in the corresponding Borel hierarchy.

\vskip1em
\noindent Let us fix a safety objective $W = \Safe(P)$ induced by $P \subseteq A^*$.
Let $w \in A^*$, define its left quotient as:
$$w^{-1} W = \set{\rho \in A^\omega \mid w \cdot \rho \in W}.$$
We denote $\Res(W)$ the set of left quotients of $W$.
We mention some special left quotients: the initial one, $\varepsilon^{-1} W$ (equal to $W$), and the empty one, obtained as $w^{-1} W$ for any $w \in P$.
As a small abuse of notations, from now on by ``left quotient'' we mean ``non-empty left quotient of $W$''.

From a left quotient $w^{-1} W$ and a letter $a \in A$, we define $(w^{-1} W) \cdot a$
as $(w \cdot a)^{-1} W$: it is easy to check that this is well defined (independent of the representant $w$ chosen).
Recall that $\Res(W)$ is finite if and only if $W$ is regular, and in such case it can be used to describe the set of states of the minimal deterministic automaton recognizing $W$.

\vskip1em
Before stating our main results, we need some order-theoretic definitions.
The width of a partially ordered set $(X,\le)$ is the cardinal of a maximal antichain of $X$ with respect to $\le$,
\textit{i.e.} the cardinal of a maximal set of pairwise incomparable elements.
Recall that an antichan is a subset of $X$ such that any two distinct elements are incomparable, and a chain is a subset of $X$ such that any two elements are comparable.
We say that $(X,\le)$ is well founded if every chain contains a minimal element.

We write $\memFD(W)$ for the quantity defined as $\mem(W)$ but restricting to arenas of finite degree: for every vertex, there are finitely many outgoing edges.

\begin{thm}\label{thm:main}
For all safety objectives $W$,
\begin{itemize}
	\item $\memFD(W)$ is the width of $(\Res(W),\subseteq)$.
	\item if $(\Res(W),\subseteq)$ is well-founded (in particular if $W$ is regular), then $\mem(W)$ is the width of $(\Res(W),\subseteq)$.
\end{itemize}
\end{thm}

We present in~\ref{sec:applications} an example, called the outbidding games, showing that the well-founded assumption is necessary.

An objective $W$ is half-positional if $\mem(W) = 1$.
In the case of safety objectives, we obtain the following characterization:

\begin{cor}
For all safety objectives $W$, $W$ is half-positional over arenas of finite degree if and only if the inclusion is a linear order over $\Res(W)$.

If furthermore $(\Res(W),\subseteq)$ is well-founded, then $W$ is half-positional (over all arenas) if and only if the inclusion is a linear order over $\Res(W)$.
\end{cor}

\section{Proofs}
\label{sec:safety}
\subsection*{A First Upper Bound}

\begin{lem}\label{lem:upper_bound_trivial}
For all safety objectives $W = \Safe(P)$, for all games $\G = (\A,W)$, if Eve has a winning strategy from an initial vertex $v_0$,
then she has a winning strategy using at most $|\Res(W)|$ memory states.

Consequently, $$\memFD(W)\ \le\ \mem(W)\ \le\ |\Res(W)|\ .$$
\end{lem}

\begin{proof}
We construct a memory structure $\M$, as follows: $\M = (\Res(W),W,\next)$,
where $\next(w^{-1} W,a) = (w^{-1} W) \cdot a$.
At any point in the game, the memory state computed by $\M$ is the current left quotient.
Let $\A$ be an arena. 
We construct the expanded arena $\A \times \M$ equipped with the coloring function $c' : E \times \Res(W) \to \set{0,1}$ defined by:
$$c'(\_,w^{-1} W) = \begin{cases}
0 & \textrm{ if } w^{-1} W = \emptyset \textrm{ (equivalently, } w \in P \textrm{)}\ ,\\
1 & \textrm{ otherwise}\ .
\end{cases}$$

We attach to $\A \times \M$ the safety condition induced by $B = \set{0}$,
giving rise to the game $\G \times \M = (\A \times \M,\Safe(B))$.
First observe that by construction, the plays in $\A \times \M$ are of the form 
$(e_0,c(e_0)^{-1} W) \cdot (e_1, c(e_0 \cdot e_1)^{-1} W) \cdots (e_k,c(e_0 \cdot e_1 \cdots e_k)^{-1} W)$,
so by definition of $c'$ a play is winning is $\G \times \M$ if and only if its projection (on the first component) is winning in $\G$.

It follows that a winning strategy for Eve in $\G$ from $v_0$ induces a winning strategy in $\G \times \M$ from $(v_0,W)$.
Now, thanks to Lemma~\ref{lem:memoryless_internal_safety}, since Eve wins in $\G \times \M$, she has a positional winning strategy.
This induces a winning strategy in $\G$ using $\M$ as memory structure, concluding the proof of Lemma~\ref{lem:upper_bound_trivial}.
\end{proof}

The game $\G \times \M$ defined above will be an important tool in the proofs to follow.
We will also rely on the following remark: assume we want to prove that a strategy $\sigma$ is winning.
Then it is enough to show that for all plays $\pi$ consistent with $\sigma$,
for all $k$, $c(\pi_k)^{-1} W \neq \emptyset$, where $\pi_k$ is the prefix of $\pi$ of length $k$.
This simple observation follows from the definition of safety objectives.

\subsection*{A Tighter Upper Bound}

The memory structure $\M$ is not optimal.
A first remark is that the empty left quotient (which exists if $W \neq A^\omega$) 
can be removed from the memory states as the game is lost.
This is why by ``left quotient'' we mean ``non-empty left quotient of $W$''.

The second remark is the following: let $L_1$ and $L_2$ two left quotients of $W$,
such that $L_1 \subseteq L_2$.
With the same notations as above, consider a vertex $v$ in the arena $\A$.
If Eve wins from $(v,L_1)$ in $\G \times \M$, then she also wins from $(v,L_2)$:
indeed, she can play as she would have played from $(v,L_1)$.
Since this ensures from $v$ that all plays are winning for $L_1$,
then a fortiori they are winning for $L_2$.

This suggests to restrict the memory states only to \textit{minimally winning} left quotients with respect to inclusion.
Two issues arise: 
\begin{itemize}
	\item which left quotients are winning depends on the current vertex,
so the semantics of a memory state can no longer be \textit{one} left quotient, 
but rather a left quotient for each possible vertex,
	\item there may not exist \textit{minimally winning} left quotients.
\end{itemize}

For the sake of presentation,  we first show how to deal with the first issue, assuming the second issue does not appear.
Specifically, in the following lemma, we assume that $\Res(W)$ is finite (\textit{i.e.} $W$ is regular),
implying the existence of minimally winning left quotients.
We will later drop this assumption.

\begin{lem}[Upper bound in the regular case]\label{lem:upper_bound_tight_regular}
Let $W$ a safety objective, we assume that $\Res(W)$ is finite.

For all games $\G = (\A,W)$, if Eve has a winning strategy from an initial vertex $v_0$,
then she has a winning strategy using at most $K$ memory states,
where $K$ is the width of $(\Res(W),\subseteq)$.
\end{lem}

The key intuition to have is the following: our goal is to construct a strategy $\sigma$ such that for a vertex $v$, we associate to every memory state $i$ a left quotient $L_i(v)$ ensuring that in a play $\pi$ consistent with $\sigma$, then $L_i(v)$ is an under-approximation of $c(\pi)^{-1} W$, meaning $L_i(v) \subseteq c(\pi)^{-1} W$.

\begin{proof}
We use the same notations as for the proof of Lemma~\ref{lem:upper_bound_trivial},
and construct a smaller memory structure together with a winning strategy using this memory structure.
In this proof, by winning we mean winning in the game $\G \times \M$.

Let $K$ be the cardinal of the maximal antichain of left quotients of $W$.
We construct a memory structure $\M^* = (\set{1,\ldots,K}, 1, \up)$,
and a strategy $\sigma$ induced by the next-move function $\next$.

Let $v$ be a vertex in $\A$.
We consider the set of minimal left quotients $L$ such that $(v,L)$ is winning.
(Here we use the finiteness of $\Res(W)$ to guarantee the existence of such left quotients.)
This is an antichain, so there are at most $K$ of them, we denote them by $L_1(v),\ldots,L_p(v)$, for some $p \le K$.
The \textit{key} property is that for every left quotient $L$ such that $(v,L)$ is winning, there exists $i$
such that $L_i(v) \subseteq L$.
Furthermore, we choose $L_1(v_0)$ such that $L_1(v_0) \subseteq W$.
(Indeed, by assumption $(v_0,W)$ is winning.)

We define the update function: $\up(i,(v,v'))$ is a $j$ such that $L_j(v') \subseteq L_i(v) \cdot c(v,v')$.
Note that in general, such a $j$ may not exist; it does exist if $(v', L_i(v) \cdot c(v,v'))$ is winning,
and we will prove that this will always be the case when playing the strategy $\sigma$.

We define the next-move function $\next$ (inducing $\sigma$). Let $v \in \VE$, and consider $(v,L_i(v))$: 
since Eve wins from there, there exists an edge $(v,v') \in E$ such that $(v', L_i(v) \cdot c(v,v'))$ is winning.
Define $\next(v,i)$ to be this $v'$.

We show that the strategy $\sigma$ is winning.
Consider a play $\pi = (v_0,v_1) \cdot (v_1,v_2) \cdots$ consistent with $\sigma$,
and $i_0 \cdot i_1 \cdots$ the sequence of memory states assumed along this play.
Denote $\pi_k$ the prefix of $\pi$ of length $k$, we prove that for all $k$, $L_{i_k}(v_k) \subseteq c(\pi_k)^{-1} W$.
Note that by definition, $(v_k,L_{i_k}(v_k))$ is winning, so $L_{i_k}(v_k) \neq \emptyset$, implying that $c(\pi_k)^{-1} W \neq \emptyset$.

We proceed by induction.
For $k = 0$, it follows from $L_1(v_0) \subseteq W$.
Let $k > 0$, the induction hypothesis is $L_{i_{k-1}}(v_{k-1}) \subseteq c(\pi_{k-1})^{-1} W$.
We distinguish two cases. 
\begin{itemize}
	\item Either $v_{k-1}$ belongs to Eve, then by construction of $\sigma$ we have that $(v_k,L_{i_{k-1}}(v_{k-1}) \cdot c(v_{k-1},v_k))$ is winning.
It follows that the update function is well defined, and $L_{i_k}(v_k) \subseteq L_{i_{k-1}}(v_{k-1}) \cdot c(v_{k-1},v_k)$,
which together with the induction hypothesis implies $L_{i_k}(v_k) \subseteq c(\pi_k)^{-1} W$.
	\item Or $v_{k-1}$ belongs to Adam. Since Adam cannot escape $\WE(\G \times \M)$,
we have that $(v_k,L_{i_{k-1}}(v_{k-1}) \cdot c(v_{k-1},v_k))$ is winning,
and the same reasoning concludes.
\end{itemize}
It follows that the strategy $\sigma$ is winning, concluding the proof of Lemma~\ref{lem:upper_bound_tight_regular}.
\end{proof}

We now get rid of the regularity assumption.
This means that for a vertex $v$, there may not be a minimal left quotient $L$ such that $(v,L)$ is winning.
We present two ways to get around this difficulty: either an assumption on $W$, or an assumption on the games.

\begin{lem}[Upper bound -- well founded assumption]\label{lem:upper_bound_tight_well_founded}
Let $W$ a safety objective, we assume that $(\Res(W),\subseteq)$ is well-founded.

For all games $\G = (\A,W)$, if Eve has a winning strategy from an initial vertex $v_0$,
then she has a winning strategy using at most $K$ memory states,
where $K$ is the width of $(\Res(W),\subseteq)$.

Consequently, $\mem(W)$ is smaller than or equal to the width of $(\Res(W),\subseteq)$.
\end{lem}

Note that if $W$ is regular, then $(\Res(W),\subseteq)$ is finite, hence well-founded.
The proof of Lemma~\ref{lem:upper_bound_tight_well_founded} is the same as for Lemma~\ref{lem:upper_bound_tight_regular}, once we note that the well founded assumption implies that for any vertex $v$, there exists a minimal left quotient $L$ such that $(v,L)$ is winning.


\begin{lem}[Upper bound -- finite degree assumption]\label{lem:upper_bound_tight_finite_degree}
Let $W$ a safety objective. 

For all games $\G = (\A,W)$ with finite degree, if Eve has a winning strategy from an initial vertex $v_0$,
then she has a winning strategy using at most $K$ memory states,
where $K$ is the width of $(\Res(W),\subseteq)$.

Consequently, $\memFD(W)$ is smaller than or equal to the width of $(\Res(W),\subseteq)$.
\end{lem}

The proof of Lemma~\ref{lem:upper_bound_tight_finite_degree} extends the proof of Lemma~\ref{lem:upper_bound_tight_regular}, getting rid of the well founded assumption. To this end, we need an equivalent formulation of the width of an ordered set $(X,\le)$: recall that it is defined as the cardinal of a maximal antichain of $X$. Thanks to Dilworth's theorem, the width is also the smallest number of maximal chains required to cover all elements.
This point of view will be useful in proving our results, and as we will discuss in the related works section, in further extending them.

Before going into the technical details, let us highlight the main difficulty of the proof.
Under the previous assumptions, for a vertex $v$, we could associate to every memory state $i$ a left quotient $L_i(v)$
and ensure that in a play $\pi$ consistent with the strategy we construct, $L_i(v)$ is an under-approximation of $c(\pi)^{-1} W$,
meaning $L_i(v) \subseteq c(\pi)^{-1} W$.
In this proof, we will not be able to associate to every memory state a single left quotient, but rather a chain of left quotients, and ensure that some left quotient in this chain is an under-approximation.

\begin{proof}
We use the same notations as for the proof of Lemma~\ref{lem:upper_bound_tight_regular},
and construct a memory structure together with a winning strategy using this memory structure.
Let $K$ be the width of $(\Res(W),\subseteq)$.
We construct a memory structure $\M^* = (\set{1,\ldots,K}, 1, \up)$, and a strategy $\sigma$ induced by the next-move function $\next$.

Let $v$ be a vertex in $\A$.
We consider the set $\W(v)$ of left quotients $L$ such that $(v,L)$ is winning.
Then $(\W(v),\subseteq)$ is an ordered set of width at most $K$: thanks to the characterisation above of the width, we can split $\W(v)$ into maximal decreasing (finite or infinite) chains of left quotients, denoted $\ell_1(v),\ldots,\ell_p(v)$, for some $p \le K$.
Note that this decomposition is not unique.
Up to renumbering, let us assume that $W \in \ell_1(v_0)$: indeed, by assumption $(v_0,W)$ is winning.

We say that $(v,\ell)$ is winning if for all $L \in \ell$, we have that $(v,L)$ is winning.
For $\ell$ a chain of left quotients and $a \in A$, we define $\ell \cdot a$ component-wise.
Note that even if $\ell$ is infinite, it may be that $\ell \cdot a$ is finite.

We define the update function: $\up(i,(v,v'))$ is a $j$ as follows.
\begin{itemize}
	\item If $\ell_i(v) \cdot c(v,v')$ is finite, denote $L \cdot c(v,v')$ its last element.
Choose $j$ such that $L \cdot c(v,v') \in \ell_j(v')$.
Note that in general, such a $j$ may not exist; it does exist if $(v', \ell_i(v) \cdot c(v,v'))$ is winning,
and we will prove that this will always be the case when playing the strategy $\sigma$.
	\item If $\ell_i(v) \cdot c(v,v')$ is infinite, then choose $j$ such that $\ell_j(v')$ 
has an infinite intersection with $\ell_i(v) \cdot c(v,v')$.
Such a $j$ exists without any assumption.
\end{itemize}

We define the next-move function $\next$ (inducing $\sigma$). Let $v \in \VE$, and consider $(v,\ell_i(v))$.
Let $L \in \ell_i(v)$, Eve wins from $(v,L)$, so there exists an edge $(v,v') \in E$ such that $(v', L \cdot c(v,v'))$ is winning,
we say that $(v,v') \in E$ is good for $L$.
Since $\W(v')$ is upward closed, if $(v,v') \in E$ is good for $L$, then it is good for every $L'$ such that $L \subseteq L'$.
We argue that there exists an edge $(v,v') \in E$ that is good for all $L \in \ell_i(v)$, 
\textit{i.e.} such that $(v',\ell_i(v) \cdot c(v,v'))$ is winning;
define $\next(v,i)$ to be this $v'$.
There are two cases:
\begin{itemize}
	\item Either $\ell_i(v)$ is finite, denote $L$ its last element.
Since $\ell_i(v)$ is decreasing, an edge good for $L$ is good for all $L' \in \ell_i(v)$.
	\item Or $\ell_i(v)$ is infinite.
The vertex $v$ has finite degree, so there exists an edge which is good for infinitely many $L \in \ell_i(v)$.
Since $\ell_i(v)$ is decreasing, it is good for all $L' \in \ell_i(v)$.
\end{itemize} 

We show that the strategy $\sigma$ is winning. 
Consider a play $\pi = (v_0,v_1) \cdot (v_1,v_2) \cdots$ consistent with $\sigma$,
and $i_0 \cdot i_1 \cdots$ the sequence of memory states assumed along this play.
Denote $\pi_k$ the prefix of $\pi$ of length $k$, we prove that for all $k$, 
there exists $L \in \ell_{i_k}(v_k)$ such that $L \subseteq c(\pi_k)^{-1} W$.
Note that by definition, $(v_k,\ell_{i_k}(v_k))$ is winning, so $c(\pi_k)^{-1} W \neq \emptyset$.

We proceed by induction.
For $k = 0$, it follows from $W \in \ell_1(v_0)$.
Let $k > 0$, 
the induction hypothesis implies the existence of $L \in \ell_{i_{k-1}}(v_{k-1})$ such that $L \subseteq c(\pi_{k-1})^{-1} W$.
We distinguish two cases, and denote $c_k = c(v_{k-1},v_k)$.
\begin{itemize}
	\item Either $v_{k-1}$ belongs to Eve, then by construction of $\sigma$ we have that $(v_k,\ell_{i_{k-1}}(v_{k-1}) \cdot c_k)$ is winning.
It follows that the update function is well defined, and:
	\begin{enumerate}
		\item If $\ell_{i_{k-1}}(v_{k-1}) \cdot c_k$ is finite,
		denote $L' \cdot c_k$ its last element, we have $L' \cdot c_k \in \ell_{i_k}(v_k)$.
		Since $L' \cdot c_k$ is the last element of $\ell_{i_{k-1}}(v_{k-1}) \cdot c_k$, it follows that $L' \subseteq L$.
		We have thus $L' \cdot c_k \subseteq L \cdot c_k$, so $L' \subseteq c(\pi_k)^{-1} W$,
		and $L' \cdot c_k \in \ell_{i_k}(v_k)$.
		\item If $\ell_{i_{k-1}}(v_{k-1}) \cdot c_k$ is infinite, $\ell_{i_k}(v_k)$ has an infinite intersection
		with $\ell_{i_{k-1}}(v_{k-1}) \cdot c_k$.
		So there exists $L' \subseteq L$ with $L' \in \ell_{i_{k-1}}(v_{k-1})$ such that $L' \cdot c_k$ is in this intersection.
		We have $L' \cdot c_k \in \ell_{i_k}(v_k)$ and $L' \subseteq c(\pi_k)^{-1} W$.
	\end{enumerate}
	\item Or $v_{k-1}$ belongs to Adam. Since Adam cannot escape $\WE(\G \times \M)$,
we have that $(v_k,\ell_{i_{k-1}}(v_{k-1}) \cdot c_k)$ is winning,
and the same reasoning concludes.
\end{itemize}
It follows that the strategy $\sigma$ is winning, concluding the proof of Lemma~\ref{lem:upper_bound_tight_finite_degree}.
\end{proof}

\subsection*{A Matching Lower Bound}

\begin{lem}[Lower bound]
For all safety objectives $W$, there exists a game $\G = (\A,W)$ with finite degree and an initial vertex $v_0$ where Eve has a winning strategy,
but no winning strategy using less than $K$ memory states, where $K$ is the width of $(\Res(W),\subseteq)$.

Consequently, $\memFD(W)$ and $\mem(W)$ are both greater than or equal to the width of $(\Res(W),\subseteq)$.
\end{lem}

\begin{proof}
Let $W$ a safety objective, and $K$ the width of $(\Res(W),\subseteq)$.
Let us consider $\set{w_1^{-1} W, \ldots, w_K^{-1} W}$ an antichain of left quotients of $W$.
Note that since there are countably many finite words, $K$ is either finite or countable infinite, but not larger.
For $i \neq j$, there exists $u_{i,j} \in A^\omega$ such that $u_{i,j} \in w_i^{-1} W$ and $u_{i,j} \notin w_j^{-1} W$.

We describe the game, illustrated in Figure~\ref{fig:lower_bound}. A play consists in three steps:
\begin{enumerate}
	\item From $v_0$ to $v'_0$: Adam chooses a word in $\set{w_1,\ldots,w_K}$;
	\item From $v'_0$, Eve chooses between $K$ options $v_1,\dots,v_K$;
	\item From $v_i$, Adam chooses a word in $\set{u_{i,j} : j \neq i}$.
\end{enumerate}
More formally: the game between $v_0$ and $v_0'$ is a finitely branching and well founded tree whose branches are labeled with the words $w_1, \dots, w_K$, and the leaves are identified with $v_0'$.
Both for the first and third step labeling an edge by a word is only for graphical representation: formally, we use a sequence of edges labeled by the corresponding sequence of colours.

We first show that Eve has a winning strategy from $v_0$, using $K$ memory states.
It consists in choosing the $i$\textsuperscript{th} option whenever Adam chooses the word $w_i$: 
whatever Adam chooses at the third step, $w_i \cdot u_{i,j} \in W$.

We now show that there exists no winning strategy using less than $K$ memory states.
Indeed, such a strategy will not comply with the above strategy and for some $i \neq j$, 
choose the $j$\textsuperscript{th} option if Adam chooses $w_i$.
Then Adam wins by playing $u_{j,i}$, since $w_i \cdot u_{j,i} \notin W$.
\end{proof}

\begin{figure}
\begin{center}
\begin{tikzpicture}[scale = 1]
\begin{scope}
\node[Adam]  at (0,2) (v0)	{$v_0$};
\node	    at (2,2) ()	{\begin{LARGE}$\vdots$\end{LARGE}};
\node[Eve] at (4,2)	 (v'0)	{$v'_0$};

\node[Adam]  at (6,4) (v1)	{$v_1$};
\node	    at (7,4) ()	{\begin{LARGE}$\cdots$\end{LARGE}};

\node[Adam]  at (6,2) (vi)	{$v_i$};

\node[Adam]  at (6,0) (vK)	{$v_K$};
\node	    at (7,0) ()	{\begin{LARGE}$\cdots$\end{LARGE}};

\node	    at (8,2) ()	{\begin{LARGE}$\vdots$\end{LARGE}};
\node       at (10,3.5) (vf1)	{};
\node       at (10,2.5) (vf2)	{};
\node       at (10,1.5) (vf3)	{};
\node       at (10,.5) (vf4)	{};
\end{scope}

\path[thick]
(v0) edge[bend left = 50] node[above] {$w_1$} (v'0)
(v0) edge[bend left = 25] node[above] {$w_2$} (v'0)
(v0) edge[bend right = 25] node[below] {$w_{K-1}$} (v'0)
(v0) edge[bend right = 50] node[below] {$w_K$} (v'0) 

(v'0) edge (v1)
(v'0) edge (vi)
(v'0) edge (vK)

(vi) edge[bend left = 20] node[above] {$u_{i,1}$} (vf1)
(vi) edge[bend left = 10] node[above] {$u_{i,2}$} (vf2)
(vi) edge[bend right = 10] node[below] {$u_{i,K-1}$} (vf3)
(vi) edge[bend right = 20] node[below] {$u_{i,K}$} (vf4) 
;
\end{tikzpicture}
\end{center}
\caption{The lower bound.}
\label{fig:lower_bound}
\end{figure}
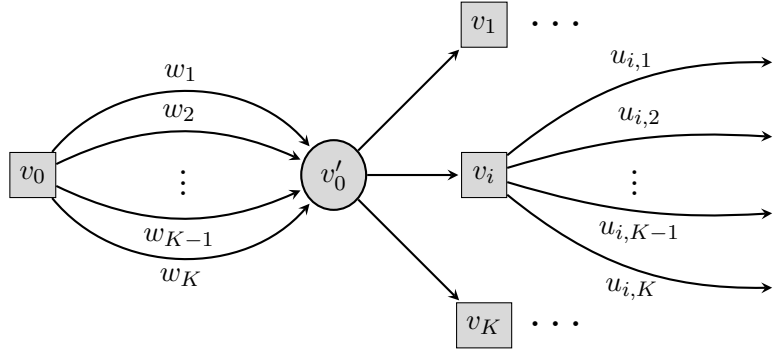

Theorem~\ref{thm:main} easily follows from combining the upper bound and lower bound lemmas.

\section{Examples and Applications}
\label{sec:applications}
In this section, we instantiate Theorem~\ref{thm:main} on different examples.
We chose four examples:
\begin{itemize}
	\item The outbidding objective shows the difference between graphs with finite degree and graphs with infinite degree;
	in particular, it gives a counter example to Lemma~\ref{lem:upper_bound_tight_finite_degree} when dropping the finite degree assumption,
	\item The energy objective is a non-regular half-positional safety objective,
	\item The generalized safety objective is a regular safety objective for which the partially ordered set of left quotients
	has a nice well-known combinatorial structure,
	\item The boundedness objective is a central piece in the theory of regular cost functions.
\end{itemize}

When representing the partial order $(\Res(W),\subseteq)$ for a given $W$,
we use the following convention: a black edge from $L$ to $L'$ means that $L \subseteq L'$,
and a dotted edge labeled $a$ from $L$ to $L'$ means that $L' = L \cdot a$,
so the dotted structure is the minimal (although possibly infinite) deterministic automaton recognizing $W$.

\subsection*{Outbidding Games}

Let $A = \set{a,b,c}$ and $W = \set{a^n \cdot b^p \cdot c^\omega \mid n \le p} \cup \set{a^\omega} \cup a^* \cdot b^\omega$.
It is a non-regular safety objective, called the outbidding objective.
The figure~\ref{fig:outbidding} represents the partial order $(\Res(W),\subseteq)$.
Its width is three: there are two incomparable infinite increasing sequences of left quotients, 
$((a^n)^{-1} W)_{n \in \N}$ and $((b \cdot a^n)^{-1} W)_{n \in \N}$,
and $c^{-1} W$.

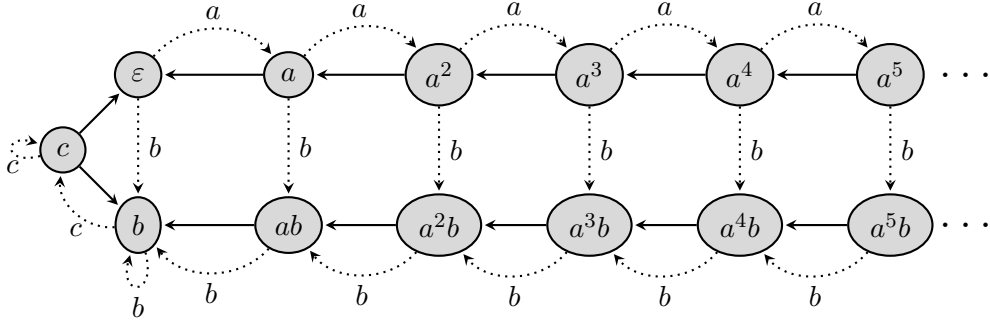
\begin{figure}[!ht]
\begin{center}
\begin{tikzpicture}[scale = 1]
\begin{scope}
\node[Eve] at (1,2)	 (eps)	{$\varepsilon$};
\node[Eve] at (3,2)	 (a)	{$a$};
\node[Eve] at (5,2)	 (aa)	{$a^2$};
\node[Eve] at (7,2)	 (aaa)	{$a^3$};
\node[Eve] at (9,2)	 (aaaa)	{$a^4$};
\node[Eve] at (11,2)	 (aaaaa)	{$a^5$};
\node      at (12,2)	 (aaaaaa)	{\begin{LARGE}$\ldots$\end{LARGE}};

\node[Eve] at (1,0)	 (b)	{$b$};
\node[Eve] at (3,0)	 (ab)	{$ab$};
\node[Eve] at (5,0)	 (aab)	{$a^2b$};
\node[Eve] at (7,0)	 (aaab)	{$a^3b$};
\node[Eve] at (9,0)	 (aaaab)	{$a^4b$};
\node[Eve] at (11,0)	 (aaaaab)	{$a^5b$};
\node      at (12,0)	 (aaaaaab)	{\begin{LARGE}$\ldots$\end{LARGE}};

\node[Eve] at (0,1)	 (c) 	{$c$};
\end{scope}

\path[thick]
(a) edge (eps)
(aa) edge (a)
(aaa) edge (aa)
(aaaa) edge (aaa)
(aaaaa) edge (aaaa)

(ab) edge (b)
(aab) edge (ab)
(aaab) edge (aab)
(aaaab) edge (aaab)
(aaaaab) edge (aaaab)

(c) edge (b)
(c) edge (eps)

(eps) edge[bend left = 50,dotted] node[above] {$a$} (a)
(a) edge[bend left = 50,dotted] node[above] {$a$} (aa)
(aa) edge[bend left = 50,dotted] node[above] {$a$} (aaa)
(aaa) edge[bend left = 50,dotted] node[above] {$a$} (aaaa)
(aaaa) edge[bend left = 50,dotted] node[above] {$a$} (aaaaa)

(eps) edge[dotted] node[right] {$b$} (b)
(a) edge[dotted] node[right] {$b$} (ab)
(aa) edge[dotted] node[right] {$b$} (aab)
(aaa) edge[dotted] node[right] {$b$} (aaab)
(aaaa) edge[dotted] node[right] {$b$} (aaaab)
(aaaaa) edge[dotted] node[right] {$b$} (aaaaab)

(ab) edge[bend left = 50,dotted] node[below] {$b$} (b)
(aab) edge[bend left = 50,dotted] node[below] {$b$} (ab)
(aaab) edge[bend left = 50,dotted] node[below] {$b$} (aab)
(aaaab) edge[bend left = 50,dotted] node[below] {$b$} (aaab)
(aaaaab) edge[bend left = 50,dotted] node[below] {$b$} (aaaab)

(b) edge[bend left = 50,dotted] node[below] {$c$} (c)
(b) edge[loop below, dotted] node[below] {$b$} () 
(c) edge[loop left, dotted] node[below] {$c$} () 
;
\end{tikzpicture}
\end{center}
\caption{The outbidding objective: more $b$'s than $a$'s.}
\label{fig:outbidding}
\end{figure}

Hence thanks to Theorem~\ref{thm:main}, $\memFD(W) = 3$.
However, there exists an outbidding game where Eve wins but needs infinite memory.
This does not contradict Theorem~\ref{thm:main}, as this game, represented in Figure~\ref{fig:infinite_memory}, has a vertex of infinite degree.
It goes as follows: first Adam picks a number $n$, and then Eve takes over: she has to pick a number $p$, higher than or equal to $n$. 
A finite memory strategy can only choose from finitely many options, hence cannot win against all strategies of Adam.

\begin{figure}
\begin{center}
\begin{tikzpicture}[scale = 1]
\begin{scope}
\node[Adam] at (0,2)	 (v0)	{$v_0$};
\node[Eve] at (2,2)	 (v'0)	{$v'_0$};
\node[Eve] at (5,4)	 (v1)	{$v_1$};
\node[Eve] at (5,3)	 (v2)	{$v_2$};
\node[Eve] at (5,2)	 (v3)	{$v_3$};
\node[Eve] at (5,.5)	 (vn)	{$v_n$};
\node      at (3.4,1.8)	 	{\begin{LARGE}$\vdots$\end{LARGE}};
\node      at (3.4,.6)	 	{\begin{LARGE}$\vdots$\end{LARGE}};
\end{scope}

\path[thick]
(v0) edge[loop above] node[above] {$a$} ()

(v0) edge[] node[above] {$a$} (v'0)

(v'0) edge[bend left = 25] node[above] {$b$} (v1)
(v'0) edge[bend left = 15] node[above] {$b^2$} (v2)
(v'0) edge[bend left = 5] node[above] {$b^3$} (v3)
(v'0) edge node[below] {$b^n$} (vn)

(v1) edge[loop right] node[right] {$c$} () 
(v2) edge[loop right] node[right] {$c$} () 
(v3) edge[loop right] node[right] {$c$} () 
(vn) edge[loop right] node[right] {$c$} () 
;
\end{tikzpicture}
\end{center}
\caption{An outbidding game with infinite degree where Eve needs infinite memory to win.}
\label{fig:infinite_memory}
\end{figure}
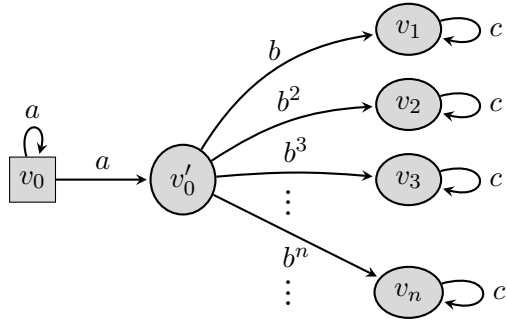

\subsection*{Energy Games}

The setup for the energy objective is the following: assume we are monitoring a resource.
We denote by $A$ the set of actions on this resource, which is any monotonic function $f : \N \to \N$, as for instance:
\begin{itemize}
	\item consuming one unit of the resource,
	\item reloading by one unit,
	\item emptying the resource,
	\item consuming half of the current energy level.
\end{itemize}
Define the energy objective by 
$$W = \set{w = w_0 w_1 \cdots \mid \textrm{ the energy level remains always non-negative}}\ .$$
It is a non-regular safety objective. 
Energy games and several variants have been extensively studied~\cite{BouyerFahrenbergLarsenMarkeySrba08,ChatterjeeDoyen12,ChakrabartiAlfaroHenzingerStoelinga03}.
The Figure~\ref{fig:energy} represents the partial order $(\Res(W),\subseteq)$,
with only two actions: $a$ reloads by one unit, and $b$ consumes one unit.
In general, if the actions are monotonic, then the left quotients are totally ordered by inclusion, so thanks to Theorem~\ref{thm:main} we have $\memFD(W) = 1$.

\begin{figure}[!ht]
\begin{center}
\begin{tikzpicture}[scale = 1]
\begin{scope}
\node[Eve] at (1,2)	 (eps)	{$\varepsilon$};
\node[Eve] at (3,2)	 (a)	{$a$};
\node[Eve] at (5,2)	 (aa)	{$a^2$};
\node[Eve] at (7,2)	 (aaa)	{$a^3$};
\node[Eve] at (9,2)	 (aaaa)	{$a^4$};
\node[Eve] at (11,2)	 (aaaaa)	{$a^5$};
\node      at (12,2)	 (aaaaaa)	{\begin{LARGE}$\ldots$\end{LARGE}};
\end{scope}

\path[thick]
(eps) edge (a)
(a) edge (aa)
(aa) edge (aaa)
(aaa) edge (aaaa)
(aaaa) edge (aaaaa)

(eps) edge[bend left = 50,dotted] node[above] {$a$} (a)
(a) edge[bend left = 50,dotted] node[above] {$a$} (aa)
(aa) edge[bend left = 50,dotted] node[above] {$a$} (aaa)
(aaa) edge[bend left = 50,dotted] node[above] {$a$} (aaaa)
(aaaa) edge[bend left = 50,dotted] node[above] {$a$} (aaaaa)

(a) edge[bend left = 50,dotted] node[below] {$b$} (eps)
(aa) edge[bend left = 50,dotted] node[below] {$b$} (a)
(aaa) edge[bend left = 50,dotted] node[below] {$b$} (aa)
(aaaa) edge[bend left = 50,dotted] node[below] {$b$} (aaa)
(aaaaa) edge[bend left = 50,dotted] node[below] {$b$} (aaaa)
;
\end{tikzpicture}
\end{center}
\caption{The energy objective: always more $a$'s than $b$'s.}
\label{fig:energy}
\end{figure}
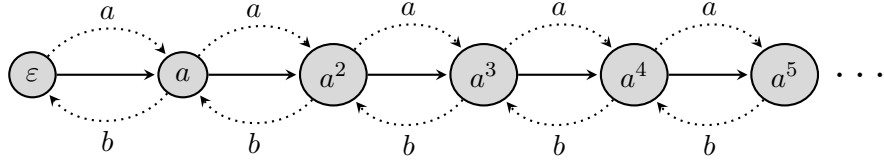

\begin{corC}[\cite{BouyerFahrenbergLarsenMarkeySrba08}]
The energy games are half-positional.
\end{corC}

\subsection*{Generalized Safety Games}

This example originates from the study of generalized reachability games~\cite{FijalkowHorn10,FijalkowHorn13}.
A generalized reachability objective is a (finite) conjunction of reachability objectives.
Here we take the opponent's vantage point: a generalized safety objective is a (finite) disjunction of (internal) safety objectives.
Specifically, let $A = \set{\bot,1,\ldots,k}$: each letter is a color, and $\bot$ is uncolored.
Let $W = \set{w = w_0 w_1 \cdots \mid \exists i \in \set{1,\ldots,k}, \forall n, w_n \neq i}$, it is satisfied if at least one color is not seen along the play.
It is a safety objective. The figure~\ref{fig:generalized_safety} represents the partial order $(\Res(W),\subseteq)$ for $k = 3$.
The left quotients are all the strict subsets of $\set{1,\ldots,k}$.
The width of this partial order is $\binom{k}{\lfloor k/2 \rfloor}$, according to the well known Sperner's Lemma from combinatorics.

\begin{figure}[!ht]
\begin{center}
\begin{tikzpicture}[scale = 1]
\begin{scope}
\node[Eve] at (1,2)	 (eps)	{$\varepsilon$};
\node[Eve] at (4,4)	 (1)	    {$\set{1}$};
\node[Eve] at (4,2)	 (2)	    {$\set{2}$};
\node[Eve] at (4,0)	 (3) 	{$\set{3}$};
\node[Eve] at (8,4)	 (12)   {$\set{1,2}$};
\node[Eve] at (8,2)	 (13)   {$\set{1,3}$};
\node[Eve] at (8,0)	 (23) 	{$\set{2,3}$};
\end{scope}

\path[thick]
(eps) edge (1)
(eps) edge (2)
(eps) edge (3)

(1) edge (12)
(2) edge (12)
(1) edge (13)
(3) edge (13)
(2) edge (23)
(3) edge (23)

(eps) edge[bend left = 25, dotted] node[above] {$1$} (1)
(eps) edge[bend left = 25, dotted] node[above] {$2$} (2)
(eps) edge[bend right = 25, dotted] node[below] {$3$} (3)

(1) edge[bend left = 25, dotted] node[above] {$2$} (12)
(2) edge[bend left = 25, dotted] node[above] {$1$} (12)
(1) edge[bend right = 25, dotted] node[below] {$3$} (13)
(3) edge[bend left = 25, dotted] node[above] {$1$} (13)
(2) edge[bend right = 25, dotted] node[below] {$3$} (23)
(3) edge[bend right = 25, dotted] node[below] {$2$} (23)

(eps) edge[loop left, dotted] node[left] {$\bot$} () 

(1) edge[loop above, dotted] node[above] {$\bot,1$} () 
(2) edge[loop below, dotted] node[below] {$\bot,2$} () 
(3) edge[loop below, dotted] node[below] {$\bot,3$} () 

(12) edge[loop above, dotted] node[above] {$\bot,1,2$} () 
(13) edge[loop below, dotted] node[below] {$\bot,1,3$} () 
(23) edge[loop below, dotted] node[below] {$\bot,2,3$} () 
;
\end{tikzpicture}
\end{center}
\caption{The generalized safety objective.}
\label{fig:generalized_safety}
\end{figure}
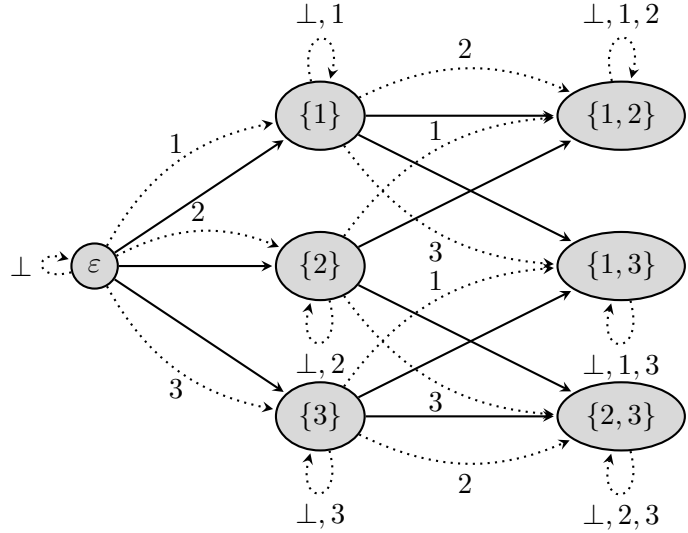

\begin{corC}[\cite{FijalkowHorn10,FijalkowHorn13}]
For all generalized safety games with $k$ colors, if Eve has a winning strategy,
then she has a winning strategy with $\binom{k}{\lfloor k/2 \rfloor}$ memory states.

Furthermore, for all $k$, there exists a generalized safety game with $k$ colors where Eve has a winning strategy using $\binom{k}{\lfloor k/2 \rfloor}$ memory states, but none using less memory states.
\end{corC}

\subsection*{Games with Counters}

This example originates from the theory of regular cost functions~\cite{Colcombet13}.
Let $N \in \N$, and define the boundedness objective $W_N$ involving a counter as follows:
$$W_N = \set{w \mid \textrm{the counter value in } w \textrm{ remains bounded by } N}\ .$$
For the set of actions, we consider any monotonic action (even non-regular), \textit{i.e.} function $f : \N \to \N$
such that if $i \le j$ then $f(i) \le f(j)$, as for instance:
\begin{itemize}
	\item leaving the counter value unchanged,
	\item incrementing the counter value by one,
	\item resetting the counter value to zero,
	\item dividing the counter value by two, rounded down,
	\item increasing the counter value to the next power of two.
\end{itemize}
The objective $W_N$ is a regular safety objective.

\begin{corC}[\cite{Colcombet13}]
The boundedness games are half-positional.
\end{corC}

\section{Related Works}
\label{sec:related_works}
The study of memory requirements in games has a long history.

\textbf{Before 2014 -- the publication of the conference version of this paper.}
The first general result in this direction is about Muller objectives: in~\cite{DziembowskiJurdzinskiWalukiewicz97},
the authors show how to compute the exact memory requirements by looking at the so-called Zielonka tree~\cite{Zielonka98}.
This is orthogonal to our results, as the Muller objectives only specify the limit behaviour 
(what is seen infinitely often), whereas we consider here only the behaviours in the finite.
Gimbert and Zielonka~\cite{GZ05} (see also~\cite{Gimbert07}) gave a characterization of all payoff functions (extending objectives to a quantitative setting) for which both players have memoryless optimal strategies over finite games.

Eryk Kopczy{\'n}ski investigated memory requirements for $\omega$-regular objectives~\cite{Kopczynski09}, asking the following question: given a $\omega$-regular objective, can we compute its memory requirement? This fascinating question that we call the Kopczy{\'n}ski programme remains open to this day. One of the many contributions of Kopczy{\'n}ski in this line of work was to introduce the notion of \textit{chromatic memory}: the memory structure only considers the sequence of colours, and not the sequence of edges.
Kopczy{\'n}ski constructed an algorithm to compute the \textit{chromatic} memory requirements of an $\omega$-regular objective, and conjectured that this coincides with general memory requirements.

\textbf{After 2014.}
Casares~\cite{Cas22} has disproved Kopczy{\'n}ski's conjecture, and in particular has shown that chromatic memory requirements are hard to compute: for Muller objectives, deciding whether there is a memory structure of size $k$ becomes NP-complete and equivalent to minimizing transition-based deterministic Rabin automata.
In this direction, Casares, Colcombet and Lehtinen~\cite{CCL22} showed that computing general memory requirements for Muller objectives is equivalent to minimizing good-for-games automata.
A result by Bouyer, Randour, and Vandenhove~\cite{BRV22} provides a link between the chromatic memory requirements of all $\omega$-regular objectives (not only Muller objectives) and their representation as transition-based deterministic parity automata, but with less tight bounds on the minimal memory structures.
Bouyer et al.~\cite{BLORV22} extended the results of Gimbert and Zielonka~\cite{GZ05} to chromatic finite memory.

A promising line of work extending the approach initiated in this paper uses the notion of universal graphs, initially introduced for understanding algorithms solving games~\cite{ColcombetFGO22}, and later used by Ohlmann for giving a characterisation of all positionally determined objectives~\cite{OhlmannThesis,Ohlmann22,Ohlmann23}. The proof of the characterisation is indeed very close to the techniques we present here: indeed universal graphs form a generic way of reducing to safety games. The characterisation has then been extended to finite memory by Casares and Ohlmann~\cite{CO23}.

The subsequent works most related to the present results are by Bouyer, Casares, Randour and Vandenhove~\cite{BouyerCRV22}, and by Bouyer, Fijalkow, Randour, and Vandenhove~\cite{BFRV23}. The first give a characterization of half-positional objectives recognized by deterministic B{\"{u}}chi automata, and the second characterizations of chromatic memory requirements for open and closed objectives, as well as complexity-theoretic results about computing these requirements. Both are technically rooted in the very ideas developed in this paper. In particular, Dilworth's theorem and the use of covering chains (developed in Lemma~\ref{lem:upper_bound_tight_finite_degree}) is at the heart of the characterization given in~\cite{BFRV23}.

Technically unrelated but close in spirit, the article~\cite{BLT22} establishes the existence of finite-memory optimal strategies from topological properties of objectives. There are major differences with our work: their framework is different (they study \emph{concurrent} games that are not played \emph{on graphs}), and their aim is to establish the existence of finite-memory optimal strategies for many objectives, but not to understand precisely the memory requirements of some class of objectives.

\section*{Conclusion and Perspectives}

We considered general safety objectives and characterized their memory requirements. Specifically, the memory requirements of a safety objective $W$ is the width of the partially ordered set $(\Res(W),\subseteq)$. This is the first general result characterizing the memory requirements for \textit{some} non-regular objectives, based on their topological properties. Back in 2014, we hoped that this would be a stepping stone for obtaining memory requirements characterizations for other classes of objectives.
As we have discussed in the related works section this hope has materialised, with many results along this line.
We now only hope that there will be more: the long-term goal is indeed to characterise memory requirements for all $\omega$-regular objectives, a research programme started by Kopczy{\'n}ski around 2006~\cite{Kopczynski09} and far from over, and even more active than ever.

\bibliographystyle{alphaurl}
\bibliography{bib}

\end{document}